\documentclass[lettersize,journal]{IEEEtran} %onecolumn
\usepackage[utf8]{inputenc}
\usepackage{color,xcolor}
\usepackage{epsfig}
\usepackage{amsmath,amsfonts,amsthm,amssymb}
\usepackage{graphics,graphicx,pgfplots,tikz}%,subfigure}
\usepackage{float}
\usepackage{algorithm,algorithmic}
\usepackage{svg}
\usepackage{epstopdf}
\usepackage{soul,csquotes,ulem}
\usepackage[hidelinks]{hyperref}
\usepackage[font=footnotesize]{subcaption}
\usepackage[font=footnotesize]{caption}
\usepackage[subtle,lists=tight,paragraphs=tight,mathspacing=tight,tracking=tight]{savetrees}
\usepackage[compact]{titlesec}

\newtheorem{lemma}{Lemma}
\newtheorem{corrolary}{Corollary}

\begin{document}
\title{Optimal~Sensing~Policy With Interference-Model Uncertainty}
	\author{
		\IEEEauthorblockN{
			Vincent Corlay, Jean-Christophe Sibel, and Nicolas Gresset
		}
		\thanks{
			The authors are with Mitsubishi Electric Research and Development Centre Europe, 35700 Rennes, France (e-mail: \{v.corlay, j.sibel, n.gresset\}@fr.merce.mee.com).
		}
	}
	\maketitle
	
%====================================================================================================================================================
\linespread{0.95}

\setlength{\abovedisplayskip}{4pt}
\setlength{\belowdisplayskip}{4pt}

\begin{abstract}
%This paper investigates the required number of sensing steps needed to acquire sufficient knowledge of an interferer's behaviour.
%More specifically, it is assumed that an interferer behaves according to a parametric model but the value of the model parameters are not known.
This paper considers a half-duplex scenario where an interferer behaves according to a parametric model but the values of the model parameters are unknown. We explore the necessary number of sensing steps to gather sufficient knowledge about the interferer's behavior.
With more sensing steps, the reliability of the model-parameter estimates is improved, thereby enabling more effective link adaptation.
However, in each time slot, the communication system experiencing interference must choose between sensing and communication.
Thus, we propose to investigate the optimal policy for maximizing the expected sum communication data rate over a finite-time communication.
This approach contrasts with most studies on interference management in the literature, which assume that the parameters of the interference model are perfectly known.
We begin by showing that the problem under consideration can be modeled within the framework of a Markov decision process (MDP). 
%We first show that the considered problem can be modeled in the Markov decision process (MDP) framework. 
Following this, we demonstrate that both the optimal open-loop and optimal closed-loop policies can be determined with reduced computational complexity compared to the standard backward-induction algorithm. 
\end{abstract}


%\tableofcontentsF
\vspace{-3mm}
\section{Introduction}

Interference-mitigation techniques have been extensively studied in the literature with a primary focus on predicting the availability of communication resources to optimize their utilization.
%This area of research} involves trying to predict the availability of the communication resources to optimize their usage.
This approach is relevant to both standard interference management and cognitive radio. 
In the context of cognitive radio, the secondary user seeks to avoid the primary user of a communication medium. 
Although the applications of these two fields are slightly different, the underlying methodologies remain similar.

Interference-mitigation methods generally involve several key steps: spectrum sensing of the communication channel, interference prediction or primary-user behavior prediction for cognitive radio cases, spectrum-access decision, and link adaptation.
For example, the optimal sensing duration to optimize the communication data rate is investigated in \cite{Liang2008} in the scope of cognitive radio. 
Similarly, \cite{Xing2014} investigates optimal sensing duration, but incorporates an interference model. More specifically, in \cite{Xing2014} sensing is conducted with few occurrences. The communication data rate is optimized as a function of the expected number of missed transmission slots or slots experiencing interference.
Regarding interference prediction and primary-user detection, Markov models are frequently used \cite{Xing2014}-\cite{Noh2010}, where the interferer activation probability at any time $t$ depends on the interferer state at time $t-1$. In many models, the interferer state is either active or inactive.
%In \cite{Akbar2007}, the prediction model is used to determine the optimal duration of the transmission in a given channel.
%Link adaptation (coding rate) is investigated in \cite{Brighente2021}. The authors focus on efficient interference prediction when the codeword is altered by the same (unknown) signal-to-interference-and-noise ratio (SINR).
Recently, reinforcement learning (RL) approaches have also been applied for this class of problems. For example, \cite{Evagoras2020} proposes a RL framework in which an agent collects observations from the environment, e.g., a user equipment and the network, and learns proper policies to adjust the estimated SINR aiming to maximize a reward function such as a user-equipment data rate. Additional resources on this topic are provided in the survey \cite{RLsurvey}.

%\textbf{Limitations of existing works.} 
Most of these studies assume that the parameters of the interference model are perfectly known and primarily focus on sensing in order to optimize the usage of the interference model.
In the half-duplex case, no analysis is made on the trade-off between: $(i)$ improving the interference-model accuracy through sensing, at the cost of missed communication opportunities, and $(ii)$ proceeding with communication based on the current partial knowledge of the model, which may lead to imperfect link adaptation.
%Moreover, it is assumed that the SINR resulting from the interference is constant over the whole codeword transmission.

\noindent \textbf{Problem statement.} In this paper, we consider a parametric interference model with no prior knowledge of the interference-model parameters. %that the number of observations available to characterize the parameters of the model is finite.
The following question is investigated:
\textit{Given a communication duration of $T$ time slots, should one allocate any given time slot to sensing or to communication in order to maximize the expected sum communication data rate?}
This applies for instance to scenarios with dynamic interference where the hidden interference parameters change every $T$ time slots.
%This scenario corresponds to a  with dynamic interference context, where \textcolor{blue}{the hidden interference parameters change} every $T$ time slots. 

%We consider a simple model where an interferer is active with probability $p$ and not active with probability $1-p$ (in an independent and identically distributed (i.i.d) manner). The parameter $p$ is unknown at the beginning of the communication.
%The extension to the two-states Markov model is left for future work. 
%(see also the conclusions).
%As a result, the reliability of the model should be taken into account.

\noindent \textbf{Structure of the paper and contributions.}
In Section~\ref{Sec_models}, we introduce the interference model and channel model and characterize the model-parameter probability distribution. 
%We show that parameter to be estimated can be modelled as Beta random variables whose distributions depend on the number of observed and total interference slots $k$ and $n$, respectively. 
We then propose adapted link-adaptation strategies, including superposition coding. Section~\ref{sec_simpleModel_sensing} models the problem as a MDP and details both efficient and optimal open-loop and closed-loop policies. Section~\ref{Sec_simu_res} presents simulation results to demonstrate the performance of the proposed algorithms.

%Section~\ref{sec_simpleModel_sensing}, we first describe how the model can be represented as an MDP.
%\textcolor{blue}{Then, we} describe optimal open-loop and closed-loop policies. 
%Section~\ref{Sec_simu_res} provides simulation results to illustrate the performance of the proposed algorithms.

\vspace{-1mm}
\section{Interference model, channel model, and link-adaptation strategies}
\label{Sec_models}
%\subsection{Settings}
\vspace{-1mm}
\subsection{Interference and sensing models}
\noindent \textbf{Interference model.} We consider a scenario where an interferer is active with probability $p$ and not active with probability $1-p$. 
At each time slot, the interferer draws a state ``active" or ``inactive" based on $p$ and remains in that state for one time slot. Then, it switches to a new state for the next time slot.
This behavior is assumed to be independent and identically distributed. \\
%The extension to the two-state Markov model, which requires the solutions for this simple model, is left for future work. \\
\textbf{Sensing model.} 
We consider a communication system that is affected by the interferer activity and is unaware of the parameter~$p$. No prior information on $p$ is assumed, i.e., it follows the uniform distribution in $[0 \ 1]$.
To estimate $p$, the system performs $n$ sensing observations and detects active interference in the channel $k$ times. Thus, the current sensing state is represented by ($k$, $n$). 
%A new sensing operation in one time slot increments $n$. The variable $k$ is incremented only if the interference state is ``active" in that time slot. 

\subsection{Interference characterization}

Generating a state $(k,n)$ is equivalent to drawing $n$ samples from a Bernoulli random variable $Y$ with a hidden parameter $p$. The corresponding binomial random variable is $Z=\sum_{i=1}^n Y_i$, where the probability of observing the interference channel on $k$ time slots out of $n$ time slots for a given parameter $p$ is $P(Z=k|n,p)=\binom{n}{k}p^k (1-p)^{n-k}$. 
As $p$ is unknown, it should be treated as a random variable. A state $(k,n)$ can be obtained by different hidden values of $p$. 
Hence, we consider the distribution $P(X=p|k,n)$ where $X$ is the random variable representing the uncertainty in $p$. Applying Bayes law and the uniform \textit{a priori} assumption on $p$, we get $P(X=p|k,n) \propto P(Z=k|p,n)$. 
To fully characterize the probability distribution of $p$, we compute the normalization constant $\int_{0}^1 \binom{n}{k}y^k (1-y)^{n-k} dy$ leading to $1/(n+1)$. Thus, the probability density function of $X$ is given by:
%\small
\begin{align}
\begin{split}
f_{k,n}(p) %&=\frac{\binom{n}{k}x^k (1-x)^{n-k}}{\int_{0}^1 \binom{n}{k}y^k (1-y)^{n-k} dy } \\
=  (n+1)\binom{n}{k}p^k (1-p)^{n-k}. % 
\end{split}
\end{align}
%\normalsize
This distribution is known as the Beta distribution \cite{BookBeta}, i.e., $X \sim Beta(k+1,n-k+1)$. The expected value is $E[X] = \frac{k+1}{n+2}$.
%The expected value of $X$ is $E[X] = \frac{k+1}{n+2}$.

In this paper, we adopt the following notations $k'=k+1$, $n'=n+1$, $k''=k+2$, and $n''=n+2$.  We also define $p_k^n = \frac{k'}{n''}$ and $\bar{p}_k^n = 1- p_k^n$.
%\vspace{-1mm}
\subsection{Considered channel}

The considered channel equation is $y=x+w+z$, where $x$ is the transmitted symbol of average power $P$, $w$ is the noise of power  $N$, $z$ is the interference of power $I$ when active, and $y$ is the received symbol. To simplify the explanation, we assume that both the noise and the interference follow a Gaussian distribution, as is common in information theory, see e.g., \cite[Sec. 15.1.5]{Cover2005}.
Without active interference, the power of the total noise is noted $N_1=N$, and with active interference, the power of the total noise is $N_2=N + I$. Therefore, the signal to interference and noise ratio (SINR) is $\mathrm{SINR}_i = P/N_i$. Typically, $\mathrm{SINR}_2<<\mathrm{SINR}_1$.   
% where $\mathrm{SINR}_1$ is the signal to interference and noise ratio (SINR) when the interferer is not active. %We call this case the first channel.
%$\mathrm{SINR}_2$ is the SINR when the interferer is active. %We call this case the second channel.
%We assume that the noise power $N$ is constant.
According to Shannon's Channel Capacity Theorem, the highest achievable communication data rate on this channel is $C(\mathrm{SINR})=(1/2) \log_2(1+\mathrm{SINR})$ [bits].  Thus, the state of the interferer induces two channels. The first channel capacity is $C(\mathrm{SINR}_1)=R_{max}$ and the second channel capacity is $C(\mathrm{SINR}_2)=R_{min}$.
We assume that the entire codeword is transmitted within one time slot, thereby experiencing a single interferer state.
%experiences the same interference state, i.e., it is transmitted within one time slot. 
Hence, the optimal communication data rate per time slot is $M \cdot C(\mathrm{SINR})$, where $M$ is the size of the codeword. To simplify the notations, this constant $M$ is omitted in the paper.

\subsection{Link adaptation}

When the communication system decides to initiate a transmission, the communication settings should be selected according to the state of the noise and the state of the interferer. This process is known as link adaptation.
In this paper, we consider the following superposition-coding strategy that employs two codes.  
It is similar to the one used for the broadcast channel \cite{Cover1972}\cite[Sec. 15.1.3]{Cover2005}, refer to Equations (15.11) and (15.12) in \cite{Cover2005} for further details.
While a formal proof of optimality is lacking, no alternative scheme has been identified that achieves a higher expected communication data rate.
On the transmitter side, the sum of the two following codewords (from two distinct codes) is transmitted:
%\vspace{-1mm}
\begin{itemize}
\item A first codeword with a power boosting constant $0\leq \alpha \leq 1$. The communication data rate $R_1$ of the first code is  $R_1(\alpha) = C(\frac{\alpha P}{N_1})$.
\item A second codeword with a boosting constant $1-\alpha$. The communication data rate $R_2$ of the second code is $R_2(\alpha) = C(\frac{(1-\alpha) P}{\alpha P + N_2})$.
%\vspace{-1mm}
\end{itemize}
On the receiver side:
\begin{itemize}
%\vspace{-1mm}
\item The receiver first attempts to decode the second codeword. It is successfully decoded with a very high probability regardless of the interference state.
\item Second, the decoded codeword is subtracted from the received signal.
\item Third, the receiver attempts to decode the first codeword. However, this succeeds only if $\mathrm{SINR} = \mathrm{SINR}_1$ as the code is designed based on the assumption of $N_1$.
%in the case of no interference. Indeed, the rate $R_1(\alpha)$ is designed assuming $N_1$. 
If $\mathrm{SINR} = \mathrm{SINR}_2$, $R_1(\alpha)$ far exceeds the channel capacity, i.e.,  $R_1(\alpha) >> C(\alpha P / N_2)$, leading to an error probability close to 1.
\end{itemize}
Thus, in the absence of interference, the communication data rate is the sum of the rates of both codes. When interference is present, the obtained data rate is that of the second code only. 
%More formerly, the rates of the two codes $R_1$ and $R_2$ are:
%\begin{align}
%, \ .
%\end{align}
Given the knowledge of $p$, the expected instantaneous communication data rate is therefore:
\vspace{-1mm}
\begin{align}
\label{equ_link_adapt}
\begin{split}
E_p[R(\alpha,p)] = R_1(\alpha)(1-p) + R_2(\alpha),
\end{split}
\end{align}
where $\alpha$ is the link-adaptation constant. 
A poor estimate of $p$ can lead to a suboptimal choice of $\alpha$, thereby reducing the communication data rate.
In the absence of perfect knowledge of $p$ but with $k,n$ observation values, the expected instantaneous communication data rate becomes:
\vspace{-1mm}
\small
\begin{align}
\label{equ_rate_good}
\begin{split}
E_f[R(\alpha,k,n)] &= \int_{0}^1 f_{k,n}(p) E_p[R(\alpha,p)] dp  = R_2(\alpha) +  R_1(\alpha)\bar{p}_k^n.
\vspace{-2mm}
\end{split}
%\vspace{-4mm}
\end{align}
\normalsize
By setting the derivative of \eqref{equ_rate_good} with respect to $\alpha$ to 0 and bounding the solution within the range $[0 \ 1]$ we derive the optimal value of $\alpha$:
\begin{align}
\alpha_k^n = \min\{\max\{0, (N_2\cdot \bar{p}_k^n - N_1)/(P \cdot  p_k^n)\},1\}.
\end{align}
Then, $E_f[R(\alpha_k^n,k,n)]$ is the expected rate with the value of $\alpha$ that maximizes \eqref{equ_rate_good}. 

The expected data rate without employing superposition coding, i.e., using only one code per transmission, can be expressed as: 
\small
\begin{align}
\begin{split}
E_f[R(\alpha_k^n,k,n)] &= \max_{\alpha} \{E_f[R(\alpha=0,k,n)], E_f[R(\alpha=1,k,n)] \},\\
&= \max\{R_{min}, R_{max}\bar{p}_k^n\}.
\end{split}
\end{align}
\normalsize
This can be interpreted as a constrained superposition-coding case, leading to inferior or equal performance. For instance, with $\mathrm{SINR}_1=15$ dB, $\mathrm{SINR}_2=10$ dB, and known $p=0.32$, the communication data rate gain achieved with superposition coding is approximately $12 \ \%$.

%\vspace{-2mm}
\section{Optimal sensing strategy}
\label{sec_simpleModel_sensing}

%\subsection{Presentation of the problem}
%\label{sec_sensing_pb}

As outlined in the introduction, we consider a half-duplex model where the system must choose between sensing and communication at each time slot to maximize the expected sum data rate over $T$ time slots. 
%On the one hand, sensing induces a missed communication slot. On the other hand, it increases the expected instantaneous communication rate for future communication slots via a better link-adaptation $\alpha$ due to a better knowledge of $p$, see \eqref{equ_link_adapt}.
In this section, we explain how to compute the optimal policies to achieve the best balance between sensing and communication.
We simplify the notation of \eqref{equ_rate_good} as $R(\alpha_k^{n},k,n)$ for $E_f[R(\alpha_k^n,k,n)]$. 
%the expected rate. With the first link adaptation strategy: %Focusing on the expected rate first, as a function of $k$ and $n$ is 
%\begin{align}
%R(k,n) = \max \{ R_{min}, R_{max} (1-\frac{k+1}{n+2}) \}.
%\end{align}

\subsection{Modeling as a MDP}

The system can be effectively modeled as a finite-horizon Markov decision process (MDP) \cite[Chap. 4]{Puterman1994}.
The states, actions, rewards, and state-transition probabilities are as follows.
\begin{itemize}
\item States: The set $(k,n,i)$ defines a state $s(k,n,i)$ of the system, where $i$ is the slot index with $0\leq k \leq n \leq i \leq T$. %Equivalently, it could be defined as a time-sequence $s_i=(k,n)$. 
%(which should not be confused with the number of sensing steps, as in the previous section). 
The number of states is given by: $\sum_{l=0}^T \sum_{m=0}^l m = O(T^3).$
%\begin{align}
%\begin{split}
%%&\sum_{i=1}^T ( i  \sum_{n=0}^i ( n  \sum_{k=0}^n k ))= O(T^6). %= \frac{1}{48} T^2(T+1)^2(T+2)(T+3)\\
%&\textcolor{blue}{ \sum_{l=0}^T \sum_{m=0}^l m = O(T^3).}
%\end{split}
%\end{align}
\item Actions: Sensing ($a=1$) or communication ($a=0$).
\item (short-term) Rewards  (which does not depend on $i$):
\begin{align}
r(k,n,a=1) = 0, \ r(k,n,a=0)=R(\alpha_k^n,k,n).
\end{align}
\item State-transition probabilities: For the action $a=0$: 
\begin{align}
\label{equ_transi_0}
p(s(k,n,i+1)|s(k,n,i),a=0)=1,
\end{align}
 and 0 for all other states. 
For the action $a=1$: 
\begin{align}
\label{equ_transi}
\begin{split}
&p(s(k+1,n+1,i+1)|s(k,n,i),a=1)=p_k^n, \\
&p(s(k,n+1,i+1) | s(k,n,i),a=1)=  \bar{p}_k^n.
\end{split}
\end{align}
\end{itemize}
It is important to note that the transition probability is $p_k^n$ rather than $p$. The purpose of the MDP formulation is to calculate the expected gain averaged over the distribution of $p$.
%, not one realization of the gain with a given value of $p$

The backward-induction (BI) algorithm can be employed to optimally solve this MDP \cite[Chap. 4]{Puterman1994}. This algorithm helps determine the optimal trade-off between improving the estimation of $p$ through sensing to increase the expected instantaneous communication data rate $R(\alpha_k^{n},k,n)$ while minimizing the number of missed transmission opportunities. \\
\textbf{Complexity of the standard implementation of BI.} 
Due to the sequential nature of the problem, the BI algorithm is applied on a trellis where at time $i$ there are $\sum_{m=0}^i m =O(i^2)$ states. Each state at time $i$ has 3 connections with the set of states at time $i+1$,  as indicated in \eqref{equ_transi_0} and \eqref{equ_transi}. Consequently, the number of edges in the trellis is given by $\sum_{l=0}^T \sum_{m=0}^l 3 m =O(T^3)$, which is the same order of magnitude as the number of states.

%However, the complexity is high due to the high number of states.
In the subsequent sections, we propose an alternative method to find the optimal strategy. The proposed closed-loop policy represents a more efficient implementation of the BI algorithm. Instead of having $O(i^2)$ states in the trellis at time $i$, we demonstrate that this can be reduced to $O(i)$.

\subsubsection{Review of MDP and notations }
We define the long-term gain starting from time $i$ as $G_i=\sum_{j=i}^T r_j$, where $r_j$ represents the short-term reward obtained at time $j$. The value of a state is defined by the expected gain starting from that state: $V(s(k,n,i))=E[G_i]$. It is equal to the expected sum communication data rate. %=E[\sum_{j=i}^T r(k,n,j) ]$. 
For convenience, we will refer to a value of a state $V(s(k,n,i))$ as $V(k,n,i)$.
Under a specific policy $\pi$, the value is denoted as $V^\pi$  where $\pi^*$ represents the optimal policy.
The notation $d_i=d(k,n,i)=a_i$ is used to denote the action taken in the state $s(k,n,i)$ under the optimal policy $\pi^* = (d_1,...d_T)$. 
Open-loop policies, where the sequence of actions does not depend on the visited states, will also be considered, i.e., $\pi =(a_1,...,a_T)$.

\subsubsection{Similarities with bandit problems}
\label{sub_sec}
The problem we are addressing shares similarities with bandit problems \cite[Chap.2]{Sutton2020} where the optimal balance between exploitation and exploration must be identified. In this context, unlike the MDP framework, the reward does not generally depend on the state (as introduced in \cite[Chap.3]{Sutton2020}).  

The following greedy bandit algorithm (\cite[Chap.2, Sec 2.2]{Sutton2020}) can be considered as a benchmark: With probability $1-\epsilon$ the algorithm selects the optimal action based on the short-term reward, while with the probability $\epsilon$ it chooses a suboptimal action (i.e., sensing in our case). Furthermore, the proposed open-loop policy can be interpreted as a bandit algorithm as the action does not depend on the state.

\subsection{Efficiently finding the optimal open and closed-loop policies}
The first two subsections present intermediary results used to obtain the open-loop policy (third subsection) and the closed-loop policy (fourth subsection).

\subsubsection{Expected communication data rate after $i$ sensing steps (open loop)}

Let us assume that the current time is $i$ and that $k$, $n$ are the current sensing values. 
The variable $k_{i+1}$ denotes the number of interference observations at time $i+1$.
The expected instantaneous rate after one additional sensing step can be expressed as:
\begin{align}
\begin{split}
E_{k_{i+1}}[R(\alpha_k^{n'},k,n')] = &p(k_{i+1}=k'|k,n) R(\alpha_{k'}^{n'},k',n')  +\\
& p(k_{i+1}=k|k,n) R(\alpha_k^{n'},k,n'), 
\end{split}
\end{align}
where $ p(k_{i+1}=k'|k,n) = p_k^n$. 
%\begin{align}
%\label{expe_rate_sens}
%\begin{split}
%&E_k[R(k,n')] =g(k,n,1) = p_k g(k',n',0)  + \bar{p}_k g(k,n',0).
%\end{split}
%\end{align}
%where $g(k,n,0) = R(\alpha_k^{n},k,n)$, $p_k = p(k_{t+1}=k'|k,n)$. %, $\bar{p}_k = 1 - p_k$. 
The expected instantaneous rate after $j$ additional sensing steps is denoted as $E_{k_{i+j}}[R(\alpha_k^{n+j},k,n+j)] = g(k,n,j)$ with:
\begin{align}
\label{expe_rate_sens}
\begin{split}
&g(k,n,j) =  p_k^n  g(k',n',j-1)  +  \bar{p}_k^n  g(k,n',j-1),
\end{split}
\end{align}
%\vspace{-2mm}
where $g(k,n,0) = R(\alpha_k^{n},k,n)$.
Equation \eqref{expe_rate_sens} corresponds to an open-loop mode where sensing is performed independently of the resulting states. 
%Note also that the coding scheme is adapted after each sensing step: at each iteration the term $R(\alpha_k^n,k,n)$ comprises the link adaptation $\alpha_k^n$ based on the current $k,n$ values, not the initial ones.
%where $g(\cdot,\cdot,0) = g_0(k,n)$.
The following lemma demonstrates that sensing always increases or maintain the expected instantaneous communication data rate.
The difference between \eqref{equ_key_1} and \eqref{equ_key_2} below highlights the advantages of sensing in this context.
\vspace{-1mm}
\begin{lemma} 
\label{lem_increa}
For any $j\geq 0$, $g(k,n,j+1) \geq g(k,n,j)$. 
\end{lemma}
\vspace{-3mm}
\begin{proof}
Consider $j=0$. If no sensing is performed:
\begin{align}
\label{equ_key_1}
\begin{split}
&E_{k_{i+1}}[R(\alpha_k^n,k,n)] = p_k^n   R(\alpha_k^{n},k',n') + \bar{p}^n_k  R(\alpha_k^{n},k,n' ),
%=   R(\alpha_k^{n},k,n) \\
\end{split}
\end{align}
indicating that the link-adaptation variable $\alpha$ is determined based on the ``old" $k$ and $n$ values. If sensing is performed:
\begin{align}
\label{equ_key_2}
E_{k_{i+1}}[R(\alpha_k^n,k,n)]  = p_k^n  R(\alpha_{k'}^{n'},k',n')+ \bar{p}_k^n   R(\alpha_{k}^{n'},k,n'),
\end{align}
where $R(\alpha_k^n, k,n)\geq R(\alpha^l_m,k,n)$ for any $l,m$ as, by definition, $\alpha_k^n$ is the best  link-adaptation value given the $k$, $n$ sensing values. 
This reasoning can be applied similarly for higher values of $j$.
\end{proof}
%The rate gain is:
%\begin{align}
%\begin{split}
%\Delta R &= E_{k_{t+1}}[E^{k_{t+1},n+1}[R] ]  - R(k,n) =\frac{k+1}{n+2}R(k+1,n+1)  + (1-\frac{k+1}{n+2}) R(k,n+1) - R(k,n),
%\end{split}
%\end{align}
%Which is always greater than 0 (to prove)
Computing \eqref{expe_rate_sens} using a recursive approach leads to an exponential complexity $O(2^j)$.
However, a more efficient method can be used.
Equation \eqref{expe_rate_sens} can be expressed as follows:
\begin{align}
\label{equ_simpli}
\begin{split}
& p_k^n  g(k',n',j-1)  +  \bar{p}_k^n  g(k,n',j-1) \\
&= p_k^n  (p_{k'}^{n'} g(k'',n'',j-2) + \bar{p}_{k'}^{n'} g(k',n'',j-2) ) \\
&+ \bar{p}_k^n   (p_{k}^{n'} g(k',n'',j-2) + \bar{p}_{k}^{n'} g(k,n'',j-2) ) \\
&= p_k^n  p_{k'}^{n'} g(k'',n'',j-2) +  \\
&(p_k^n \bar{p}_{k'}^{n'} + \bar{p}_k^n  p_{k}^{n'})g(k',n'',j-2) +\bar{p}_k^n \bar{p}_{k}^{n'} g(k,n'',j-2).
\end{split}
\end{align}
where $g(k',n'',j-2)$ needs to be computed only once rather than twice as required in the recursive approach.
%\begin{align}
%\begin{split}
%&= p_k \cdot p_{k+1} g(k'',n'',i-2) +  \\
%&(p_k \bar{p}_{k'} + \bar{p}_k  p_{k})g(k',n'',i-2) +\bar{p}_{k} g(k,n'',i-2).
%\end{split}
%\end{align}
More generally, at level $j-m$ there are only $m+1$ distinct functions $g$ to evaluate after the merging step (corresponding to the last row of \eqref{equ_simpli}): At this level the value of $k$ ranges between $k$ and $k+m$.
Algorithm~\ref{algo_eff} outlines an efficient algorithm to compute $g(k,n,j)$ with a complexity of $O(\sum_{l=0}^j l) = O(j^2)$. 
%The complexity at each level is therefore $O(j)$.
\begin{algorithm}
\caption{Computing the expected instantaneous communication data rate after $j$ open-loop sensing steps $E_{k_{i+j}}[R(\alpha_k^{n+j},k,n+j)]=g(k,n,j)$}
\begin{algorithmic}[1]
\label{algo_eff}
\STATE $p^{vec}$ = getProbaVec($k,n,j$)
\STATE $g(k,n,j) = \sum_{k_l=k}^{k+j} p^{vec}_{k_l} R(\alpha_{k_l}^{n+j},k_l,n+j)$%g(k_l,n+j,0)$. 
\end{algorithmic}
\vspace{3mm}
\begin{algorithmic}[1]
\STATE \textbf{Function $p^{vec}$ = getProbaVec(k,n,j)}
\STATE Initiate $p^{vec}$ as a size 1 vector with a value 1.
\FOR {$l=j:1$}
%\STATE $\%$At level $l=j-m$
\STATE (At level $l=j-m$) Get $p^{vec}$ of size $m+1$ which represents the probabilities of having $k_l=k,k_l=k+1,...,k_l=k+m$ (calculated in step 6).
\STATE Increment $n$ and for each value of $k_l$ multiply its corresponding probability by $p_{k_l}^n$ and $\bar{p}_{k_l}^n$, which correspond to the values $k_l+1$ and $k_l$, respectively. This yields a vector of probabilities of size $2(m+1)$.
\STATE Merge the vector to obtain $p^{vec}$ with $m+2$ proba. values.
\ENDFOR
\STATE Return $p^{vec}$.
\end{algorithmic}
\end{algorithm}

%In Appendix~\ref{app_alte_way}, we explain why directly expressing $g(k,n,j)$ as a function of $g(.,.,0)$ does not yield an efficient algorithm.

%Assume that at a given level iteration we have $K$ values of $k$ and $K$ corresponding probabilities:
%\begin{itemize}
%\item For all current value $k$, get $2K$ values $k$ and $k+1$.
%\item Compute the probabilities corresponding to each obtained value.
%\item Merge the same $k$ values by adding their probabilities.
%\end{itemize}

\subsubsection{The on/off property of the problem}

The following lemma shows that once sensing is stopped, the action communication should always be chosen. 
This is a logical conclusion as the observation values $k$, $n$ remain unchanged, while the available time to collect the gain diminishes.
%do not change but the remaining time to collect the gain (sum rate) diminishes. 
\vspace{-1mm}
\begin{lemma}
\label{lemma_all_0}
Let $(a_1,...,a_T)$ represent a realization of an optimal policy. If $a_i=0$, then $a_j=0$, for any $j>i$. 
\end{lemma}
\vspace{-1mm}
\begin{proof}
According to Theorem 4.3.3 in \cite{Puterman1994}, the optimal policy is derived by solving the optimality/Bellman equation by determining the state values that satisfy the equation. Actions are then selected based on this optimality equation, which is expressed as follows in our context:
\begin{align}
\label{equ_value}
\begin{split}
&V(k,n,i)= \max \{R(\alpha_k^n, k,n) +  V(k,n,i+1), \\
&p_k^n V(k',n',i+1) + \bar{p}_k^n V(k,n',i+1) \}. 
%a^* = \text{arg max}_{a} \{ r(k,n,a) + p_k V(k+1,n+1,T) + \bar{p}_k V(k,n+1,T)\}
\end{split}
\end{align}
Let us assume that $a_{i-1}=d(k,n,i-1)=0$. It implies that: 
\begin{align}
\begin{split}
&R(\alpha_k^n, k,n)+ V(k,n,i) > p_k^n V(k',n',i) + \bar{p}_k^n V(k,n',i) \\
&\geq p_k^n Q(k',n',i,a=0) + \bar{p}_k^n Q(k,n',i,a=0) \\
%&= p_k R(k',n') + \bar{p}_k R(k,n') + \\
%& p_k V(k',n',i +1) + \bar{p}_k V(k,n',i+1) \\
&= E_{k_i}[R(\alpha_k^{n'},k,n')] +  p_k^n V(k',n',i +1) + \bar{p}_k^n V(k,n',i+1),
\end{split}
\end{align}
where $Q(\cdot)$ is the standard $Q$-value and where $E_{k_i}[R(k,n')]\geq  R(\alpha_k^n, k,n)$ as sensing always increases the expected short-term rate (Lemma~\ref{lem_increa}). 
To develop the left-hand term, we assume that $a_{i}=d(k,n,i)=1$ (the opposite of the lemma for $j=i+1$):
\begin{align}
\begin{split}
&R(\alpha_k^n, k,n) + V(k,n,i)  \\
&=R(\alpha_k^n, k,n) + p_k^n V(k',n',i +1) + \bar{p}_k^n V(k,n',i+1) \\
&> E_{k_i}[R(\alpha_k^{n'},k,n')] + p_k^n V(k',n',i +1) + \bar{p}_k^n V(k,n',i+1),  \\
\end{split}
\end{align}
which yields $R(\alpha_k^n, k,n) > E_{k_i}[R(\alpha_k^{n'},k,n')]$. This leads to a contradiction since $E_{k_i}[R(\alpha_k^{n'},k,n')]\geq  R(\alpha_k^n, k,n)$.
Thus, the policy $(a_1,..,a_{i-1}=0,a_i=0, a_{i+1},...)$ is superior to $(a_1,..,a_{i-1}=0,a_i=1, a_{i+1},...)$ given that $a_{i-1}=0$.
Furthermore, by applying the same argument, we can demonstrate that  $(a_1,...,a_{i-1}=0 ,a_i=0, a_{i+1}=0,...)$ is better than $(a_1,..,a_{i-1}=0, a_i=0, a_{i+1}=1,...)$. Applying this reasoning iteratively completes the proof of the lemma. 
\end{proof}

\subsubsection{The optimal open-loop policy}

This on/off property allows us to establish the following lemma. 
We recall that $g(k,n,j)$ can be efficiently computed with Algorithm~\ref{algo_eff}.
\vspace{-2mm}
\begin{lemma}
\label{lemma_onoff}
%To determine whether one should communicate or sense at slot $i$ (provided that $a_{i-1}=1$): 
Let the current time be $i-1$. Consider the following $(T-i)$ open-loop policies: $\pi_1 = (a_i=0,0, ... )$, $\pi_2=(a_i=1,0,0, ... )$,   $\pi_3 = (a_i=1,1,0, ... )$, $\pi_4 = (a_i=1,1,1,0, ... )$, etc.  
The corresponding expected gain for each policy is $g(k,n,0) (T-i)$, ..., $g(k,n,j) (T-i-j)$, ..., respectively. 
If at least one of the above gains $g(k,n,j) (T-i-j) >g(k,n,0) (T-i)=(T-i) R(\alpha_k^n,k,n)$, for $j>0$, then the optimal action at time $i$ is sensing.
 \end{lemma}
However, it is important to note that the lemma does not specify the optimal action if $(T-i) R(\alpha_k^n,k,n)$ is the largest value. %(i.e., it does not mean that communication is optimal).
The expected gains in Lemma~\ref{lemma_onoff} represent the values of the state $s(k,n,i)$ under the candidate open-loop policies. %if one needs to set in advance the number of sensing time slots to be performed.
Therefore, the highest gain $g(k,n,j) (T-i-j)$ of the $j$-th open-loop policy can serve as an optimal open-loop policy. 
%As mentioned in Subsection~\ref{sub_sec}, this open-loop policy can be seen as an efficient bandit policy.
\begin{corrolary}
\label{ref_coro}
Let the current time be $i-1$. The optimal open-loop policy is determined by selecting the number of sensing steps to $j$, $0\leq j \leq (T-i)$, such that $g(k,n,j) (T-i-j)$ has the highest value.
\end{corrolary}
\begin{proof} (of Lemma~\ref{lemma_onoff})
The equation for $V(k,n,i)$ is given by \eqref{equ_value}.
%\begin{align}
%\label{equ_value}
%\begin{split}
%&V(k,n,i)= \max \{ p_k^n V(k',n',i) + \bar{p}_k^n V(k,n',i), \\
%& R(k,n) + V(k,n,i+1) \}. 
%\end{split}
%%V(k,n,i)= \max \{ p_k V(k+1,n+1,i) + \bar{p}_k V(k,n+1,i),  (T\times i)  R(k,n) \}
%%a^* = \text{arg max}_{a} \{ r(k,n,a) + p_k V(k+1,n+1,T) + \bar{p}_k V(k,n+1,T)\}
%\end{align}
If $d(k,n,i)=0$, the maximum of \eqref{equ_value} is obtained from the first term in \eqref{equ_value}. By applying Lemma~\ref{lemma_all_0}, we have $Q(k,n,i,a=0) = (T- i) R(\alpha_k^n,k,n)$.
Therefore, if at least one policy $\pi$ exists with $d(k,n,i)=1$, that produces a value $V^{\pi}(k,n,i)>(T- i) R(\alpha_k^n,k,n)$, then the optimal action is $d(k,n,i)=1$. 
%The second term of \eqref{equ_value} can be lower bounded by any open-loop policy containing $d(k,n,i)=1$. 
%Focusing on $V(k',n',i+1)$, we have:
%\begin{align}
%\begin{split}
%& V(k',n',i+1) = \max \{ p_{k'}^{n'} V(k'',n'',i+2) \\ 
%&+ \bar{p}_{k'}^{n'} V(k',n'',i+2),    (T- i-1) R(\alpha_{k'}^{n'}, k',n')\} \\
%& \geq p_{k'}^{n'} V(k'',n'',i+2) + \bar{p}_{k'}^{n'} V(k',n'',i+2).
%\end{split}
%\end{align}
%Hence,  for any $0\leq j \leq (T-i)$:
%\begin{align}
%\begin{split}
%&p_k^n V(k',n',i+1) + \bar{p}_k^n V(k,n',i+1)  \\
%&\geq [p_k^n g(k',n', j)+  \bar{p}_k^n g(k,n',j)](T - i-1-j) .
%\end{split}
%\end{align}
Consequently, if for any $j \geq 0$, $g(k,n,j) (T-i-j) \geq (T- i) R(\alpha_k^n,k,n)$, then we find: 
\begin{align}
V(k,n,i) \geq  g(k,n,j)(T-i-j) \geq   (T- i) R(\alpha_k^n, k,n),
%V(k,n,i)= \max \{ p_k V(k+1,n+1,i) + \bar{p}_k V(k,n+1,i),  (T\times i)  R(k,n) \}
%a^* = \text{arg max}_{a} \{ r(k,n,a) + p_k V(k+1,n+1,T) + \bar{p}_k V(k,n+1,T)\}
\end{align}
where $V(k,n,i)$ is obtained with $d(k,n,i)=1$.
\end{proof}

%\begin{lemma}
%XXX
%\end{lemma}

\subsubsection{The optimal closed-loop policy}

The value of the states under the optimal closed-loop policy can also be computed efficiently using an algorithm of similar complexity.
The key idea is to introduce a pre-processing part before the standard BI algorithm, allowing the algorithm to be executed only on the subset of relevant states.
Let $f(a,b,c)=\max\{a+b,c\}$.
We can re-express $V(k,n,i)$, as given by \eqref{equ_value}, as follows: %$V(k,n,i) = f(p_k^n a_i,\bar{p}_k^n b_i, c_i)$ 
\begin{align}
\begin{split}
&V(k,n,i) = f(p_k^n a_i,\bar{p}_k^n b_i, c_i), 
\end{split}
\end{align}
where $a_i=V(k',n',i), \ b_i=V(k,n',i)$, and $c_i=(T-i) R(\alpha_k^n,k,n)$.
The terms $a_i$ and $b_i$ are unknown and therefore must be further developed.
We let $i'=i+1$.
\begin{align}
\begin{split}
&a_i=f(p_{k'}^{n'} a_{i'}, \bar{p}_{k'}^{n'} b_{i'},c_{i'}), \ b_i=f(p_k^{n'}d_{i'},\bar{p}_k^{n'} e_{i'},f_{i'}).
\end{split}
\end{align}
Similarly, $c_{i'}$ and $f_{i'}$ can be directly computed using Lemma~\ref{lemma_all_0}.
The terms $a_{i'}$, $b_{i'}$, $d_{i'}$, and $e_{i'}$ must be further expanded.
However, $b_{i'} = d_{i'} = V(k',n'',i+2)$ similarly to the observation made in \eqref{equ_simpli}.
Therefore, only three terms require further expansion: $a_{i'}, b_{i'}, e_{i'}$.
The connections between these two levels can be stored as a vector $[1, \ 2, \ 2, \ 3]$ where the value at position $j$ indicates the index of the next-level state to be used.
More generally, at level $i+j$ there are $j$ ``fixed terms" to store and the connection vector has the form $[1, \ 2, \ 2, \ 3, \ 3 , ... , j+1]$.
Once level $i+j=T$ is reached, the algorithm proceeds in the other direction using the standard BI algorithm starting from the obtained terminal states.

Algorithm~\ref{alg:algo_eff_2} summarizes the two parts.
%The first part is used to identify the useful states and useful connections.
Regarding the complexity, similarly to Algorithm~\ref{algo_eff}, there are $O(j)$ operations to perform per level.
Thus, the overall complexity is $O(\sum_{l=i}^T l ) = O(T^2 - i^2)$.

\begin{algorithm}
\caption{Computing the value of $V(k,n,i)$ under the optimal closed-loop policy}\label{alg:algo_eff_2}
\begin{algorithmic}[1]
\label{algo_value_states}
\STATE Begin with the current $k$,$n$ values.
\STATE \textbf{First part} (pre-processing)
\FOR{$l=i:T$}
\STATE (At level $l=i+j$) Store the $j+1$ fixed terms, each computed as $(T-l) R(\alpha_k^n, k,n)$.
\STATE Store the connections between this level and the next one. %Identify the $j+1$ distinct terms to expend. 
\ENDFOR
\STATE \textbf{Second part} (backward induction)
\FOR{$l=T:i$}
\STATE (At level $l=i+j$) Compute the values of the $j$ states at this level using:
\begin{itemize}
\item The state values and the connection vector from the previous level  $i+j+1$.
\item The fixed terms computed in the first part. 
\end{itemize}
\ENDFOR
\end{algorithmic}
\end{algorithm}
%We proceed similarly to the backward in duction algorithm.
%Starting from any viable state $s(k,n,T)$. 
%We have $V^*(\cdot,\cdot,T)=0$.
%Then,

%To compare the open-loop and closed loop policy, we can implement an algorithm similar to Algorithm~\ref{algo_value_states}, i.e., use the function getProbaVec() to weight the value of the states instead of the $g(.,.,0)$ terms.

%The weighting coefficients $c_{k_l}$ remain unchanged but the terms $g(k,n,0)$ at step 5  is replaced by the value of the states computed by Algorithm~\ref{algo_value_states}.

\vspace{-2mm}
\section{Simulation results}
\label{Sec_simu_res}
%\subsection{Open-loop policy}

Figure \ref{fig:open_loop_poly_T2000} illustrates the values of the state $s(k=0,n=0,i=0)$ under the following policies, with and without superposition coding.
\begin{itemize}
\item The optimal open-loop policy computed with Algorithm~\ref{algo_eff}: For the $j$-th open-loop policy, the value is computed as $g(k,n,j) (T-i-j)$. As noted in Corollary~\ref{ref_coro}, the highest value determines the optimal number of open-loop sensing steps. 
\item The optimal closed-loop policy computed with Algorithm~\ref{algo_value_states}: It is presented for several maximum number of allowed sensing steps. %, represented by the flat-curves for the two link-adaptation strategies. %represent the value of the state under the optimal policy computed with algorithm~\ref{algo_value_states}.
\item The greedy bandit benchmark algorithm with $\epsilon^*=0.01$.
\item Known $p$: In this case the optimal policy is to always communicate. The expected sum rate is $T\cdot  \int_0^1 p(p) E_p[\alpha^*,p]dp$, where $p(p)$ follows the uniform distribution and $\alpha^*$ is computed using the known value $p$. 
\end{itemize}
The values of the simulation parameters are as follows: $T=1000$, $\mathrm{SINR}_1=10$ dB, $\mathrm{SINR}_2=15$ dB.

Note that the value of the state $s(k=0,n=0,i=0)$ represents the expected sum communication data rate at the end of the communication period. 
Hence, it directly provides the average sum data rate obtained when simulating the system across various interference probability values.

We begin by analyzing the coding schemes. We observe that superposition coding results in a 3 $\%$ to 4 $\%$ data rate increase compared to link adaptation without superposition. 
% for both policies with $T=400$ and for the open-loop policy for $T=2000$. For $T=2000$ and with the closed-loop policy, the relative gain of superposition coding is significantly reduced.
%We then observe that superposition coding is more robust to the absence of sensing. 
 %Moreover, for both link-adaptation strategies, the optimal number of open-loop sensing steps is around 5, which brings a 9$\%$ and 2$\%$ rate improvement over the no-sensing case, without and with superposition-coding link adaptation, respectively. Superposition coding is therefore more robust to the absence of sensing.
Additionally, the optimal value of the state with the closed-loop policy is closely approximated by the best open-loop policy. %(but with a small gap in the case of no superposition coding). 
Furthermore, the proposed open-loop policy outperforms the benchmark greedy policy by 2-3 $\%$.
 
Note however that for ``high" values of $p$, the link adaptation becomes trivial (i.e., $R_{min}$ is to be chosen).
Therefore, we focus on evaluating performance for more challenging values of $p$ restricted for instance to the range $p \in [0.05,0.3]$.
%Even though the proposed policies could be improved by taking this information into account, the same algorithms are used.
Table~\ref{table_perf_norma} presents the results with the same parameters, with no superposition coding, and with $T=500$. 
As expected, the performance gap between the policies widens. A 3 $\%$ to 4 $\%$ expected sum data rate difference is observed. %Notice also that $\epsilon^*$ is increased compared to the case $p \in [0,1]$.

%This is convenient as the complexity of computing the $j$-th open loop policy is $O(j^2)<<O(T^2)$.

%\begin{figure}
%    \centering
%    \includegraphics[scale=0.7]{open_loop_eval_400.eps}
%    \caption{$T=400$. Value of the state $s(k=0,n=1,i=1)$ under several policies. O-L stands for open loop and C-L closed loop.}
%    \label{fig:open_loop_poly_T400}
%\vspace{-4mm}
%\end{figure}
%\vspace{-1mm}
\begin{figure}
    \centering
    \includegraphics[scale=0.6]{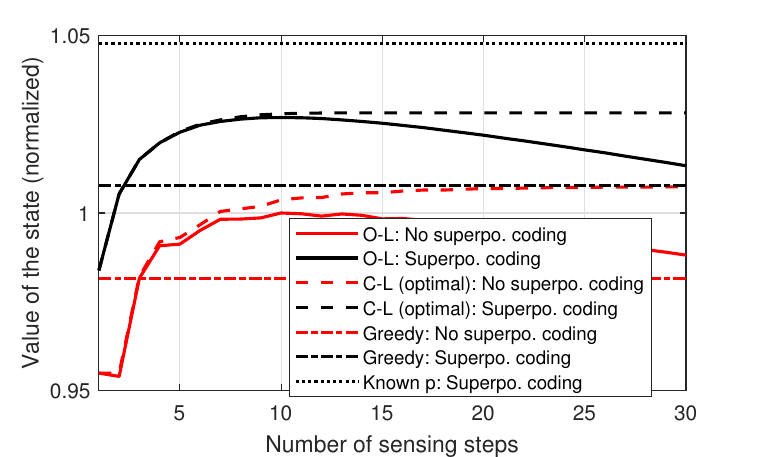}
\vspace{-1mm}
    \caption{Value of the state $s(k=0,n=0,i=0)$ under several policies and with $T=1000$. For the closed-loop policy the x-axis (number of sensing steps) represents the maximum number of allowed sensing steps. O-L refers to the open-loop policy while C-L refers to the closed-loop policy.}
%O-L and C-L refers to the open-loop policy and the closed-loop policy, respectively.}
    \label{fig:open_loop_poly_T2000}
\vspace{-2mm}
\end{figure}
\begin{table}
\begin{center}
\begin{tabular}{||c | c | c ||} 
 \hline
 Greedy ($\epsilon^*=0.03$)& Open-loop & Closed-loop   \\ [0.5ex] 
 \hline\hline
 1 & 1.04 & 1.07 \\ 
 \hline
\end{tabular}
\caption{Normalized expected sum communication data rate with $p \in [0.05,0.3]$, $T=500$, and no superposition coding.}
\label{table_perf_norma}
\end{center}
\vspace{-8mm}
\end{table}

%Of course, these values highly depend on the parameter values. For instance, a higher $n$ value decreases the sensing gain as there is already a good confidence in the interferer behaviour. Higher $T$ value decreases the cost of sensing.

Finally, Figure~\ref{fig:System_eval} illustrates an example of the evolution of a system over time with $T=1000$. 
%Hence, the true hidden value of $p$ is used to obtain the interference state in this case (not $p_k^n$).
The first subplot depicts the expected gain of each visited state over time, the second subplot displays the chosen action, the third subplot presents the current sum data rate, and the last subplot shows the current interference-state realization. We observe that the closed-loop algorithm transitions from sensing to communication at time $i=27$, and the coding rate $R_{max}$ is selected (example without superposition coding). Note that for more extreme values of $p$ (close to 0 or 1), which polarize $k$ with high probability, the sensing period is significantly shortened, typically to fewer than 5 time slots. %Unsurprisingly, for balanced values of $p$ as in the example, the algorithm takes more time to decide the best coding rate.  
One possible interpretation is that, for balanced value of $p$, a higher reliability in the estimate is necessary due to the increased potential for data rate loss associated with suboptimal link-adaptation choices.
\vspace{-1mm}
\begin{figure}
    \centering
    \includegraphics[scale=0.71]{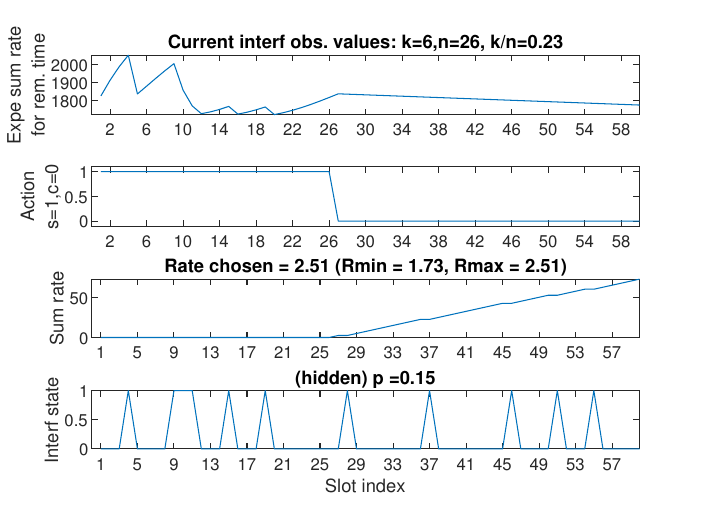}
    \caption{Example of evolution of a system over time with $T=1000$. Only the first 60 slots are displayed.}
    \label{fig:System_eval}
\vspace{-6mm}
\end{figure}

\section{Conclusions}

In this paper, we addressed a decision-making problem where the communication system must choose between communication and sensing. Sensing improves the accuracy of link adaptation, but comes at the cost of missing communication opportunities. 
We formulated this problem as a MDP and demonstrated that optimal policies, both in open and closed-loop configurations, can be derived with a low computational complexity. 
Simulation results show that the optimal closed-loop policy can be effectively approximated by the optimal open-loop policy. 

\end{document}